\documentclass[12pt]{article}

\usepackage{latexsym,amssymb,amsfonts,amsmath,amsthm,braket}
\usepackage{ytableau}
\usepackage{enumitem}
\usepackage{graphicx}
\usepackage{comment}
\usepackage{subcaption}
\usepackage{ulem}

\newtheorem{theorem}{Theorem}
\newtheorem{lemma}{Lemma}
\newtheorem{corollary}{Corollary}
\newtheorem{proposition}{Proposition}
\newtheorem{remark}{Remark}

\numberwithin{lemma}{section}
\numberwithin{proposition}{section}
\numberwithin{remark}{section}
\numberwithin{definition}{section}

\usepackage{color}
\newcommand{\red}{\textcolor{black}}

\title{Pulsation of quantum walk between two arbitrary graphs with weakly connected bridge}
\author{Taisuke Hosaka$^{1}$\thanks{Corresponding Author: hosaka-taisuke-pn@ynu.jp },  Etsuo Segawa$^1$\\
$^1${\small Graduate School of Environment and Information Sciences, Yokohama National University,}\\ 
{\small Hodogaya, Yokohama 240-8501, Japan,
}\\}
\date{}

\begin{document}
\maketitle
\centerline{\textbf{Abstract}}
We consider the Grover walk on a finite graph composed of two arbitrary simple graphs connected by one edge, referred to as a bridge. 
The parameter $\epsilon>0$ assigned at the bridge represents the strength of connectivity: if $\epsilon=0$, then the graph is completely separated. We show that for sufficiently small values of $\epsilon$, a phenomenon called pulsation occurs. The pulsation is characterized by the periodic transfer of the quantum walker between the two graphs. An asymptotic expression with respect to small $\epsilon$ for the probability of finding the walker on either of the two graphs is derived. This expression reveals that the pulsation depends solely on the number of edges in each graph, regardless of their detailed adjacency information or the position of the bridge between them. In addition, we obtain that the quantum walker is transferred periodically between the two graphs, with a period of order $O(\epsilon^{-1/2})$. Furthermore, when the number of edges of two graphs is equal, the quantum walker is almost completely transferred.


\section{Introduction}
    Quantum walks play key roles in the quantum computing like quantum algorithms \cite{Aed,Am1,G}, quantum simulation \cite{QS}, and quantum cryptography \cite{QC}.
    Such effectiveness of application of quantum walks is based on distinctive properties compared to classical random walks, such as localization \cite{IS,IKS,ILS}, ballistic spreading \cite{K} and periodicity \cite{PB, KSTY,HKSSp}.
    One of the main topics in the study of quantum walks is quantum search algorithms \cite{Portugal} and perfect state transfer \cite{Gpst1}. Spatial search algorithms based on quantum walks aim to locate the marked vertex with high probability, providing quadratic speed-up over classical search algorithms on some graphs \cite{ADZ, QCcoupling2}. Perfect state transfer refers to the phenomenon in which a quantum walker is transferred from one position to another with probability one \cite{Gpst2, KS, KY}. Characterizing graphs that allow perfect state transfer remains an important research direction. Also, without being obsessed with perfectly transferring the quantum walker, some studies have shown that a sufficient amount of quantum walkers transfer to one position, such as the antipodal point or self-loop \cite{SSstate}.
    
    Both of these topics share a common goal: concentrating the quantum walker at a specific positions. Motivated by this, we investigate a related phenomenon called pulsation, where the quantum walker repeatedly transfers between two regions of the graph \cite{TS}. This behavior may be seen as a generalization of both spatial search algorithms and perfect state transfer. In this work, we aim to understand when pulsation occurs and to clarify the key factors that cause it. We hope that this pulsation can be added as one of the distinctive properties of quantum walks.

    Our study treats the Grover walk on the finite graph constructed by two graphs connected by an edge, referred to as {\it a bridge}.
    We put a weight parametrized by $\epsilon>0$ on a bridge, while all edges except a bridge are put weight $1$. The parameter $\epsilon$ is considered as the strength of connectivity. If $\epsilon=1$, standard Grover walk on the graphs is reproduced. While if $\epsilon=0$, effectively the bridge disappears, and the graph is decomposed into the three disconnected parts: the two subgraphs and the bridge edge itself.
    This means that the smaller the value of $\epsilon$, the weaker the connectivity between two graphs in our model. Then, intuitively, if $\epsilon$ is close to 0, it would be expected that the quantum walker would have difficulty transferring between the two graphs, \red{either remaining trapped within one side of the graph or requiring an excessively long time to reach the other.
    Indeed,
    for correspondence classical random walk, it takes a large amount of time to converge to the stationary distribution, which is proportional to $\epsilon^{-1}$}.
    However, we observe a counterintuitive phenomenon: for sufficiently small $\epsilon$, the quantum walker goes back and forth between subgraphs via {\it a \red{weak} bridge} \red{on a relatively short timescale which is proportional to $\epsilon^{-1/2}$}, that is, the pulsation occurs in these settings \red{with the  period $O(\epsilon^{-1/2})$}. 
    
    \red{
    In the previous paper \cite{TS}, we demonstrated the existence of the pulsation on the graph identifying a fixed vertex of the Johnson graph with the center vertex of a star graph. Consequently, the class of the graphs we considered was limited.
    \red{
    On the other hand, in this work, by introducing a weak bridge between two arbitrary graphs, we calculate their pulsation response (see Theorem~\ref{thm: prob}). This yields a more comprehensive result.}
    }
    
    We obtain an asymptotic behavior of the finding probability of the quantum walker in either graph with respect to small $\epsilon$, presented in Theorem \ref{thm: prob} (see Section 4). Theorem \ref{thm: prob} implies that the behavior of the quantum walker determined solely by the number of edges in each graph and is independent of their \red{adjacency information or the placement of the bridge between them}. Theorem \ref{thm: periodicity} gives the periodicity of the behavior of the quantum walk of order $O(\epsilon^{-1/2})$. 
    This work provides a new feature of quantum walks called pulsation and we hope that it contributes to a deeper understanding of quantum transport phenomena.
    
    This paper is organized as follows. In Section $2$, we proposed the settings of the graphs and our quantum walk model. Section $3$ shows some numerical simulations. In Section $4$, we give our main results. In Section $5$, we address the proof of the main theorem. Section $6$ summarizes our results and discusses future work.
\section{Setting the model}
 Let $H_j=(V_j, A_j) \,(j=1,2)$ be a simple connected graph. Here, $V_j$ and $A_j$ are a set of vertices and symmetric arcs, respectively. The origin and terminal vertex of an arc $a$ are indicated by $o(a)$ and $t(a)$. The inverse arcs of an arc $a$ are denoted by $\bar{a}$. Note that $t(a)=o(\bar{a})$ and $o(a)=t(\bar{a})$ hold. 
 For a fixed vertex $\xi_j \in V_j$, we define $\bold{e}_{*}=\{(\xi_1, \xi_2), (\xi_2, \xi_1)\}$. It is called {\it the bridge} between $H_1$ and $H_2$.
 We set a graph $G=(V,A)$ where $V=V_1 \cup V_2, A=A_1 \cup A_2 \cup \bold{e_*}$. Let $\delta V=\{\xi_1, \xi_2\} \subset V$ be the boundary vertex set.
The weight function $w: {A} \rightarrow [0, 1]$ is defined by
\begin{align}
    \label{eq: weight}
    w(a)=
    \begin{cases}
        1 &: a \notin \bold{e}_*, \\
        \epsilon &: a \in \bold{e}_*.
    \end{cases}
\end{align}
We set a function $p_{\epsilon}: {A} \rightarrow [0, 1]$ as $p_{\epsilon}(a)=w(a)/m(o(a))$,
where $m(x)= \sum_{a:\, o(a)=x} w(a)$.
We should remark that $p_{\epsilon}(a)$ is regarded as the probability of a random walker from a vertex $o(a)$ to $t(a)$ depending on the weight.
Combining with a definition of $p_{\epsilon}$ and Eq.\,(\ref{eq: weight}), it follows that 
 \begin{align}\label{eq:pe}
     p_{\epsilon}(a)=
     \begin{cases}
         1/\deg(o(a)) &: o(a) \notin \delta V, \\
         1/(\deg(o(a))+\epsilon) &: o(a) \in \delta V, a \notin \bold{e}_*, \\
         \epsilon/(\deg(o(a))+\epsilon) &: a \in \bold{e}_*.
     \end{cases}
 \end{align}
Here, $\deg(x)$ is the degree given by
\begin{align*}
    \deg (x)=\left|\{a \in A_{1} \cup A_{2}\,|\, t(a) =x \}\right|
\end{align*}
for $x \in V$. Note that the bridge is not counted as a degree.
The boundary matrix $d_{\epsilon}: \mathbb{C}^{A} \rightarrow \mathbb{C}^{V}$ is defined by
\begin{align*}
    (d_{\epsilon})_{x, a}=
    \begin{cases}
        \sqrt{p_{\epsilon}(\bar{a})} &: t(a)=x, \\
        0 &: otherwise.
    \end{cases}
\end{align*}
for any $a \in A, x \in V$.
The shift matrix $S: \mathbb{C}^{A}\rightarrow \mathbb{C}^{A}$ is defined by
\begin{align*}
    (S)_{b,a}=\delta_{a,\bar{b}}.
\end{align*}
Here, $\delta_{x,y}$ is the Kronecker delta.
The time evolution matrix $U(\epsilon) : \mathbb{C}^{A}\rightarrow \mathbb{C}^{A}$ is
\begin{align*}
U(\epsilon)=S(2d_{\epsilon}^{*}d_{\epsilon}-I)
\end{align*}
\red{We note that if $\epsilon=0$, the time evolution is completely divided into the Grover walks on the disjoint graphs. Here, the time evolution of the Grover walk is given by
\[ U(0)|a\rangle=\frac{2}{\mathrm{deg}(t(a))}\sum_{b\in A:o(b)=t(a)}|b\rangle - |\bar{a}\rangle.\;\;(a\in A)  \]
This means that the time evolution operator of the Grover walk is uniquely determined by  only the underlying graph, that is, if an incident wave to a vertex $x$ transmits to the other directions, the associated weight is $2/\mathrm{deg}(x)$, if it reflects to the opposite direction, then the associated weight is $2/\mathrm{deg}(x)-1$, which is independent of the labeling of arcs. 
}
\red{
While the definition of the walk can naturally be extended by, for example,  varying the weight function $w$, this extended edge weight $w$, which depends on the labeling of the arc of course, is added to the information required to determine the time evolution.
\red{In this paper, we focus on only the graph structures as the factor of the pulsation for the first step and clarify its fundamental properties using  the Grover based coin.}
}
\begin{remark}
    If $\epsilon=0$, $U(\epsilon)$ acts on each of the three divided graphs $H_1$, $H_2$ and $\bold{e}_{*}$. While if $\epsilon=1$, $U(\epsilon)$ is regarded as Grover walks on $G$.
\end{remark}
This remark implies that the smaller $\epsilon$ is, the weaker the connection of $\bold{e}_{*}$.
In this paper, we focus on the case with sufficiently small $\epsilon$.
\red{Let $\ket{\psi_{t}}$ be the $t$-th iteration of the time evolution matrix for \red{the} initial state $\ket{\psi_{0}}$ \red{($||\psi_0||=1$)}, that is, 
\begin{align*}
    \ket{\psi_{t+1}}=U(\epsilon)\ket{\psi_{t}}=(U(\epsilon))^{t+1}\ket{\psi_{0}}.
\end{align*}
}
\red{The} initial state $\ket{\psi_0}$ is defined as the uniform superposition on $H_1$, that is,
\begin{align}
    \label{eq: initial state}
    \ket{\psi_0}=\frac{1}{\sqrt{|A_1|}}\sum_{a\in A_1}\ket{a}.
\end{align}
Let $\mu_{t}(H_j)$ be the probability of the existence of a quantum walker on $H_j\,(j=0,1,2)$ at time step $t$, denoted by
\begin{align}\label{eq:mut}
    \mu_{t}(H_j)=\sum_{a\in A_j}\left|\bra{a}U(\epsilon)^t\ket{\psi_0}\right|^2.
\end{align}
Here, we set $H_0=(\delta V, \bold{e}_*)$.
We should remark that it follows that
\begin{align*}
    \mu_{t}(H_0)=1-\mu_{t}(H_1)-\mu_{t}(H_2).
\end{align*}
 Thus, our main purpose is to estimate the asymptotic expression of $\mu_{t}(H_1)$ and $\mu_t(H_2)$ with respect to the parameter $\epsilon \ll1$, which is the strength of the connectivity between $H_{1}$ and $H_{2}$.
\section{Demonstration}


\red{
We are interested in the differences of the behavior between the random walk and quantum walk. 
It is well known in the probability theory that if the underlying graph is not bipartite, the corresponding random walk converges to the stationary distribution $\pi:X\to [0,1]$:  
\[ \pi(x)=\frac{\sum_{t(a)=x}w(a)}{\sum_{a\in A}w(a)}. \]
\red{This implies that the probability of finding a random walker in graph $H_j$ $(j=1,2)$ at time $t$, $\mu_t^{RW}(H_j)$, converges to 
\[ \lim_{t\to^\infty}\mu_t^{RW}(H_j)=\frac{\sum_{a\in A_j}w(a)}{\sum_{a\in A}w(a)}\;\;(j=1,2).  \]
}  
Then, we want to know what influences the behavior of the corresponding quantum walk by observing  $\mu_t(H_j)\; (j=1,2)$, such as the structure of the graph.
Hence, this section shows the numerical simulation in some cases.
Figure \ref{fig: k5_k5} shows the cases of random walk and quantum walk with $H_1=H_2=K_5$. Here, $K_{n}$ is the complete graph whose number of vertices is $n$. 
Figure \ref{fig: k5_k5} (a) implies that the random walk transfers the walker from $H_{1}$ to $H_{2}$ with only around half of the walkers being transported in a long time step (approximately 6,000 steps). See Remark~\ref{rem:rwconverge} for a rough estimation for its convergence time.  
In contrast, figure \ref{fig: k5_k5} (b) shows that the quantum walk can transport almost all walkers to the opposite side in a very short time (approximately 70 steps). See Theorem~\ref{thm: periodicity} for an estimation of its period. This demonstrates that random and quantum walks behave significantly differently.
}
Next, we consider the case where $H_1$ and $H_2$ are different.
The case of $H_1=K_5, H_2=K_3$ and reverse is illustrated in Figure \ref{fig: k5_k3 or k3_k_5}.
For both cases, the behavior of $\mu_t(H_2)$ remains unchanged, while that of $\mu_t(H_1)$ has changed significantly.
In particular, Figure \ref{fig: k5_k3 or k3_k_5} (b) says that there exists a time step when $\mu_t(H_1)\sim 0$.
As you can see, the behaviors of $\mu_t(H_j)$ vary depending on the graph.
We were able to obtain these asymptotically behaviors of $\mu_t(H_j)$ without \red{detailed adjacency information of each graph}, see the next section.
\begin{figure}[]
  \centering
  \begin{minipage}{0.43\columnwidth}
    \centering
    \includegraphics[width=\columnwidth]{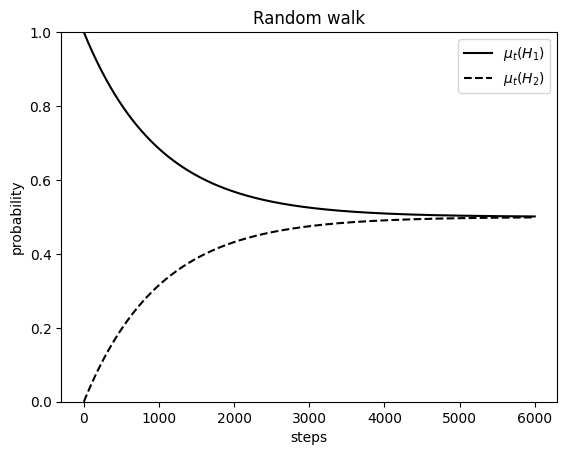}
    \subcaption{\red{Random walk case}}
  \end{minipage}
  \hspace{5mm}
  \begin{minipage}{0.43\columnwidth}
    \centering
    \includegraphics[width=\columnwidth]{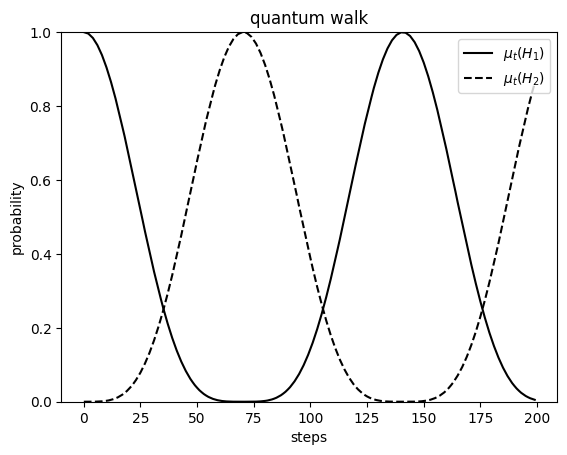}
    \subcaption{Quantum walk case}
  \end{minipage}
  \caption{\red{The horizontal and vertical axes are the time steps $t$ and the finding probability in each $H_j$ ($j=1,2$), respectively}. The solid and doted curves correspond to the finding probability on $\mu_t(H_1)$ and $\mu_t(H_2)$ where $H_{1}=H_{2}=K_{5}$ with $\epsilon=0.01$. (a) and (b) shows the random walk and quantum walk case, respectively}
  \label{fig: k5_k5}
\end{figure}

\begin{figure}[]
  \centering
  \begin{minipage}{0.43\columnwidth}
    \centering
    \includegraphics[width=\columnwidth]{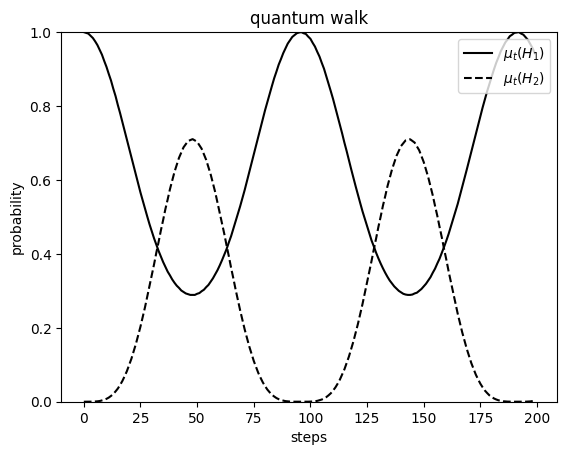}
    \subcaption{$H_1=K_5, H_2=K_3$}
  \end{minipage}
  \hspace{5mm}
  \begin{minipage}{0.43\columnwidth}
    \centering
    \includegraphics[width=\columnwidth]{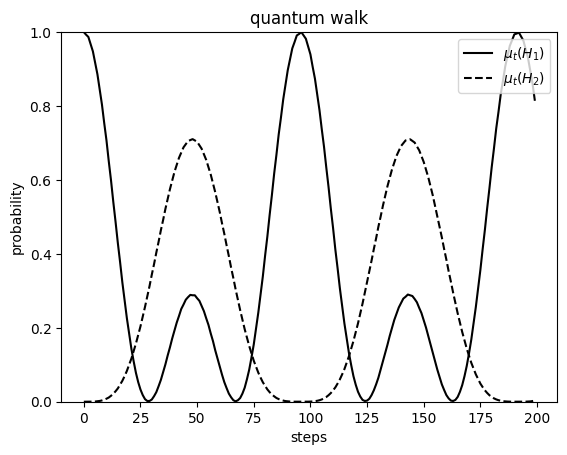}
    \subcaption{$H_1=K_3, H_2=K_5$}
  \end{minipage}
  \caption{\red{ The horizontal and vertical axes are the time steps $t$ and the finding probability in each $H_j$ ($j=1,2$), respectively.} The solid and doted curves correspond to the finding probability on $\mu_t(H_1)$ and $\mu_t(H_2)$ with $\epsilon=0.01$, respectively. }
  \label{fig: k5_k3 or k3_k_5}
\end{figure}


\section{Main Theorem}
This section shows the main theorems mentioned in the previous section.
\begin{theorem}
\label{thm: prob}
\red{For any connected graph $H_j$ ($j=1,2$), let $\mu_t(H_j)$ be defined as (\ref{eq:mut}). }
\red{Then,} for sufficiently small \red{bridge parameter between $H_1$ and $H_2$}, $\epsilon$, it \red{follows} that
    \begin{align*}
        &\mu_{t}(H_1)= \left( \frac{|A_1|+|A_{2}| \cos{(t\theta(\epsilon))}}{|A_1|+|A_2|}\right)^2 +O(\epsilon), \\
        &\mu_{t}(H_2)= \left( \frac{\sqrt{|A_1||A_2|}}{|A_1|+|A_2|}(1-\cos(t\theta(\epsilon)))\right)^2 +O(\epsilon).
    \end{align*}
    where $\theta(\epsilon)$ is the declination of the eigenvalue of  $U(\epsilon)$ that has the largest real part except $1$ given by
    \begin{align*}
        \cos \theta(\epsilon)=1-\left(\frac{1}{|A_1|}+\frac{1}{|A_2|}\right)\epsilon+O(\epsilon^2).
    \end{align*}
\end{theorem}
Here, 
Theorem \ref{thm: prob} implies that the pulsation depends only on the number of arcs, \red{which is independent of detailed adjacency information of each graph and also the connected place of the bridge between them}. Thus, it can be considered that this phenomenon is a universal property of \red{Grover walk} on finite graphs.
    
\begin{theorem}
    \label{thm: periodicity}
    Let $\tau(\epsilon)$ be the time step when $\mu_{t}(H_2)$ is maximized for the first time.
    For sufficiently small $\epsilon$, it \red{holds} that
    \begin{align*}
        \tau(\epsilon)=\left\lfloor\frac{\pi}{\sqrt{2}}\sqrt{R_{\rm eff}(|A_1|,|A_2|)} \times\epsilon^{-1/2} \right\rfloor,
    \end{align*}
    where $R_{\rm{eff}}(|A_1|,||A_2|)$ is the effective resistance of an electric circuit consisting of two resistors of resistance values $|A_1|$ and $|A_2|$ connected in parallel, that is,  $R_{\rm eff}^{-1}(|A_1|,|A_2|)=1/|A_1|+1/|A_2|$.
\end{theorem}
Theorem \ref{thm: periodicity} implies that $\tau (\epsilon)$ is order $\epsilon^{-1/2}$. In addition, it is interesting that $R_{\rm{eff}}(|A_1|,||A_2|)$, used in electric circuits, has been found. However, the relationship is not clearly clarified.
Focusing on Theorem \ref{thm: prob}, we get the condition when $\mu_t(H_2)$ obtains the maximum value.
\begin{corollary}
    \label{cor: same}
    When $|A_1|=|A_2|$ and $t=\tau(\epsilon)$, $\mu_t(H_2)$ gets the maximum value $1+O(\epsilon)$. 
    Especially in the case $|A_1|=|A_2|$, it follows that
    \begin{align*}
        \mu_t(H_1)&=\cos^4 \left(\frac{t\theta(\epsilon)}{2}\right)+O(\epsilon), \\
        \mu_t(H_2)&=\sin^4 \left(\frac{t\theta(\epsilon)}{2}\right)+O(\epsilon).
    \end{align*}
\end{corollary}
Corollary \ref{cor: same} shows that almost all quantum walker transfers from $H_1$ to $H_2$ when the number of arcs is the same on both graphs.
Figure \ref{fig: k6_c15} shows the case $K_6$ and $C_{15}$. Here, $C_n$ is the cycle graph with $n$ vertices. The structure of the graphs and the number of vertices are quite different; however, the number of arcs is equal. Therefore, almost all quantum walker come back and forth between the two graphs.
\begin{figure}[h]
    \centering
    \includegraphics[width=0.5\linewidth]{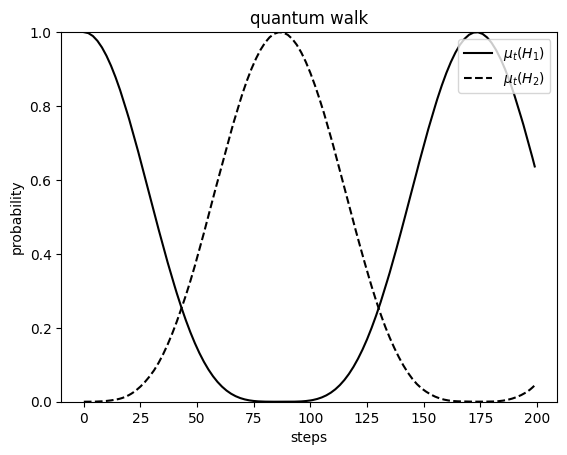}
    \caption{\red{The horizontal and vertical axes are the time steps $t$ and each finding probability in $H_j$ ($j=1,2$), respectively.} The solid and doted curves correspond to $\mu_t(K_6)$ and $\mu_t(C_{15})$ with $\epsilon=0.01$, respectively.}
    \label{fig: k6_c15}
\end{figure}

\section{Proof of Main Theorem}
In this section, we show a proof of Theorem \ref{thm: prob} and Theorem \ref{thm: periodicity}. First, we consider the spectral analysis of the random walk \red{on $G$ whose transition on each vertex to its neighbor at each time step is determined by $p_\epsilon$ in (\ref{eq:pe})}. Let $W(\epsilon)$ be this transition probability matrix of $G$ depending on the weight $w$, that is, 
\begin{align*}
    (W(\epsilon))_{x,y}=
    \begin{cases}
        1/\deg(x) &: x,y \notin \delta V, x\sim y, \\
        1/(\deg(x)+\epsilon) &: x \in \delta V, y \notin \delta V, x \sim y,\\
        \epsilon/(\deg(x)+\epsilon) &: x,y \in \delta V, \\
        0 &: otherwise,
    \end{cases}
\end{align*}
where $x \sim y$ means that the vertices $x$ and $y$ are \red{adjacent}.
Because of the property of the transition matrix, we immediately see
\begin{align}
    \label{eq: e_vec of epsilon}
    \ket{\boldsymbol{1}_V} \in \mathrm{Ker}(W(\epsilon)-I)
\end{align}
where $\ket{\boldsymbol{1}_{V}}$ is the all-one vector with respect to the set of \red{vertices} $V$. 
By an expansion with respect to $\epsilon$, $W(\epsilon)$ is denoted by
\begin{align*}
    W(\epsilon)=W^{(0)}+\epsilon W^{(1)}+O(\epsilon^{2}),
\end{align*}
where
\begin{align*}
    (W^{(0)})_{x,y}=
    \begin{cases}
        1/\deg(x) &: x\sim y, (x,y) \notin \bold{e}_{*}, \\
        0 &: otherwise,
    \end{cases}
\end{align*}
\begin{align*}
    (W^{(1)})_{x,y}=
    \begin{cases}
        -1/\deg(x)^{2} &: x \in \delta V, y \notin \delta V,   x\sim y,\\
        1/\deg(x) &: (x,y) \in \bold{e}_{*},\\
        0 &: otherwise.
    \end{cases}
\end{align*}
For simplicity, we deal with $W(\epsilon)$ expressed as a block matrix consisting of $H_1$ and $H_2$ components.
Then, $W^{(0)}$ is given by
\begin{align*}
    W^{(0)}=
    \begin{pmatrix}
        {\huge{W_1}} & {\huge{O}} \\
        {\huge{O}} & {\huge{W_2}}
    \end{pmatrix}.
\end{align*}
We should remark that $W_j \, \red{(j=1,2)}$ is the transition probability matrix of simple random walk on $H_j$. 
From the property of $W_j$, we immediately get
\begin{align*}
    \begin{bmatrix}
        \ket{\boldsymbol{1}_{V_1}} \\
        \boldsymbol{0}
    \end{bmatrix},
    \begin{bmatrix}
        \boldsymbol{0} \\
        \ket{\boldsymbol{1}_{V_2}}
    \end{bmatrix}
    \in \mathrm{Ker}(W^{(0)}-I).
\end{align*}
In other words, $W^{(0)}$ has eigenvalue 1 with multiplicity 2.
Then, the eigenvalue $1$ splits into two eigenvalues for small $\epsilon$.
Therefore, we should use a method called reduction process \cite{Kato} to get the perturbed eigenvalue of $W(\epsilon)$.
We should remark that one eigenvalue of $W(\epsilon)$ is $1$ from Eq. (\ref{eq: e_vec of epsilon}) and another one eigenvalue of $W(\epsilon)$ is expressed as 
\begin{align}\label{eq:expand}    \lambda(\epsilon)=1+\epsilon\lambda^{(1)}+O(\epsilon^2).
\end{align}
In this case, \red{according to \cite{Kato},} we should consider the following matrix
\begin{align*}
    \tilde{W}(\epsilon)=\frac{1}{\epsilon}(W(\epsilon)-I)\Pi,
\end{align*}
where $\Pi$ is the projection of $W^{(0)}$ corresponding to eigenvalue 1.
Since $1$ is semi-simple eigenvalue of $W^{(0)}$, the matrix $\tilde{W}(\epsilon)$ can be expressed as
\begin{align*}
    \tilde{W}(\epsilon)=\tilde{W}^{(1)}+O(\epsilon),
\end{align*}
where $\Tilde{W}^{(1)}=\Pi W^{(1)} \Pi$ [(2.20) in Ch. ii, Sect. 2.2 \cite{Kato}].
From a property of $W_j$, we have
\begin{align*}
    \Pi=
    \begin{bmatrix}
        \Pi_1 & {\huge O} \\
        {\huge O} & \Pi_2
    \end{bmatrix},
\end{align*}
where
\begin{align*}
    \Pi_j=\ket{\boldsymbol{1}_{V_j}}\bra{\pi_j},
\end{align*}
and $\ket{\pi_j}$ is the reversible measure of $W_{j}$ written as
\begin{align*}
    \pi_{j}(x)=\braket{\pi_j|x}=\frac{\deg(x)}{|A_j|}
\end{align*}
for $x \in V_j$.
By direct calculation, $\tilde{W}^{(1)}$ is denoted by
\begin{align*}
    \tilde{W}^{(1)}&=
    \begin{bmatrix}
        -\displaystyle \frac{1}{|A_1|}\Pi_1 & \displaystyle\frac{1}{|A_1|}\ket{\boldsymbol{1}_{V_1}}\bra{\pi_2} \\
        \displaystyle\frac{1}{|A_2|}\ket{\boldsymbol{1}_{V_2}}\bra{\pi_1} & -\displaystyle\frac{1}{|A_2|}\Pi_2
    \end{bmatrix} 
    =
    \begin{bmatrix}
        &\displaystyle\frac{1}{|A_1|}\ket{\boldsymbol{1}_{V_1}} \\
        &-\displaystyle\frac{1}{|A_2|}\ket{\boldsymbol{1}_{V_2}}
    \end{bmatrix}
    \begin{bmatrix}
        -\bra{\pi_1} & \bra{\pi_2}
    \end{bmatrix}.
\end{align*}
From \cite{Kato} and the above formula, the coefficient $\lambda^{(1)}$ can be obtained by
\begin{align*}
    \lambda^{(1)} \in \mathrm{Spec}\left(\tilde{W}^{(1)}|_{\mathrm{Ran(\Pi)}}\right)=\left\{-\frac{1}{|A_1|}-\frac{1}{|A_2|}\right\}.
\end{align*}
The eigenvector of $\lambda^{(1)}$ is calculated by
\begin{align*}
    \begin{bmatrix}
            &-|A_2|\,\ket{\boldsymbol{1}_{V_1}} \\
            &|A_1|\, \ket{\boldsymbol{1}_{V_2}}
        \end{bmatrix}
        \in \mathrm{Ker}(\tilde
        W^{(1)}-\lambda^{(1)} I).
\end{align*}
Here, we should remark that if the eigenprojections of $W(\epsilon)$ and $\tilde{W}^{(1)}$ is defined $P(\epsilon)$ and $P$, then it follows that
\begin{align*}
    P(\epsilon)=P+O(\epsilon)
\end{align*}
\red{because perturbed operator $U(\epsilon)$ is also unitary, that is, semi-simple.}
Therefore, we get the following lemma.
\begin{lemma}
    \label{lem: eigenpair of W}
    Let $W(\epsilon)$ be \red{the} transition matrix depending on the weight $w$. Then, it follows that
    \begin{align*}
        \begin{bmatrix}
            \ket{\boldsymbol{1}_{V_1}} \\
            \ket{\boldsymbol{1}_{V_2}}
        \end{bmatrix}
        \in \mathrm{Ker}(W(\epsilon)-I), 
        \quad
        \begin{bmatrix}
            &-|A_2|\,\ket{\boldsymbol{1}_{V_1}} \\
            &|A_1|\, \ket{\boldsymbol{1}_{V_2}}
        \end{bmatrix}+O(\epsilon)
        \in \mathrm{Ker}(W(\epsilon)-\cos (\theta(\epsilon)) I),
    \end{align*}
    where
    \begin{align*}
        \cos (\theta(\epsilon))=1-\left(\frac{1}{|A_1|}+\frac{1}{|A_2|}\right)\epsilon+O(\epsilon^2).
    \end{align*}
\end{lemma}
\red{
\begin{remark}\label{rem:rwconverge}
From Lemma~\ref{lem: eigenpair of W}, the mixing time of the random walk $W(\epsilon)$ can be roughly estimated by considering $\cos^t(\theta(\epsilon))<\delta$ for a small value $\delta$ because $\cos\theta(\epsilon)$ is the second largest eigenvalue of $W(\epsilon)$. 
This implies that a convergence time for the random walk is estimated by 
$O(\epsilon^{-1})$ because 
\begin{align*} 
\cos^t\theta(\epsilon) &\sim (1-(1/|A_1|+1/|A_2|)\epsilon)^t \\ 
&\sim e^{-(1/|A_1|+1/|A_2|)\epsilon t} \\
& <e^{-|\log \delta|}. 
\end{align*}
\end{remark}
}
Next, we show the proposition connecting the eigenvalue and the eigenvector between $W(\epsilon)$ and $U(\epsilon)$.
Let $\pi_{G} \in \mathbb{C}^{V}$ be the reversible measure of $W^{(0)}$, that is,
\begin{align*}
    \pi_{G}(x)=\frac{\deg(x)}{|A_{1}|+|A_{2}|}.
\end{align*}
We define the diagonal matrix $D^{1/2}: \mathbb{C}^{V}\rightarrow \mathbb{C}^{V}$ by
\begin{align*}
    (D^{1/2})_{x,y}=
    \begin{cases}
        \sqrt{\pi_{G}(x)} &:x=y, \\
        0 &:otherwise.
    \end{cases}
\end{align*}
Then, we give the following proposition.
\begin{proposition}[\cite{HKSS}]
    \label{prop: SMT}
    Let us set $f \in \mathrm{Ker}(\cos \theta (\epsilon)-W(\epsilon))$. Then, the corresponding eigenvector $\psi_{\pm \theta} \in \mathrm{Ker}(e^{\pm i \theta (\epsilon)}-U(\epsilon))$ is given by
    \begin{align*}
        \psi_{\pm \theta}(a)=
        \begin{cases}
            \displaystyle\frac{1}{\sqrt{2(|A_{1}|+|A_{2}|)}|\sin \theta|}(f(t(a))-e^{\pm i \theta} f(o(a)))+O(\sqrt{\epsilon}) &: \cos \theta \neq \pm 1, \\
            f(t(a)) &: \cos \theta=\pm 1,
        \end{cases}
    \end{align*}
    with $||D^{1/2}f||=1$.
\end{proposition}
\begin{proof}
    Let us set $f \in \mathrm{Ker}(\cos\theta (\epsilon)-W(\epsilon)),\; g \in \mathrm{Ker}(\cos \theta (\epsilon)-W_{sym})$ with $||g||=1$. Here, $W_{sym}\red{(\epsilon)}=d_{\epsilon}Sd_{\epsilon}^{*}$.
    We should remark that $W_{sym}\red{(\epsilon)}=D^{1/2}(W\red{(\epsilon)})D^{-1/2}$. Then, we immediately see $g=D^{1/2}f$, that is, 
    \begin{align}
        \label{eq: gf}
        g(x)= \sqrt{\pi_{G}(x)}f(x)
    \end{align}
    for $x \in V$.
    By spectral mapping theorem of quantum walk \cite{HKSS}, we see
    \begin{align*}
        \psi_{\pm \theta}(a)=
        \begin{cases}
            \displaystyle\frac{1}{\sqrt{2}|\sin \theta|}\left(\frac{1}{\sqrt{\deg (t(a))}}g(t(a))- \frac{e^{\pm i \theta}}{\sqrt{\deg (o(a))}}g(o(a))\right)+O(\sqrt{\epsilon}) &: \cos \theta \neq \pm 1, \\
            \displaystyle\frac{1}{\sqrt{\deg (t(a))}}f(t(a)) &: \cos \theta=\pm 1.
        \end{cases}
    \end{align*}
    Combining Eq.\,(\ref{eq: gf}) and above equation, for $\cos \theta \neq 1$ case, it is obtained by
    \begin{align*}
        \psi_{\pm \theta}(a)&=\displaystyle\frac{1}{\sqrt{2}|\sin \theta|}\left(\frac{1}{\sqrt{\deg (t(a))}}\sqrt{\pi_{G}(t(a))}f(t(a))- \frac{e^{\pm i \theta}}{\sqrt{\deg (o(a))}}\sqrt{\pi_{G}(o(a))}f(o(a))\right)+O(\sqrt{\epsilon}) \\
        &=\displaystyle\frac{1}{\sqrt{2(|A_{1}|+|A_{2}|)}|\sin \theta|}\left(f(t(a))-e^{\pm i \theta}f(o(a))\right)+O(\sqrt{\epsilon})
    \end{align*}
    with $||g||=||D^{1/2}f||=1$. Similarly, we see the $\cos \theta= 1$ case.
    Hence, we get the desired conclusion.
\end{proof}
Combining with Lemma \ref{lem: eigenpair of W} and Proposition \ref{prop: SMT}, we immediately get the following lemma.
\begin{lemma}
    \label{lem: eigenpair of U}
    Let $U(\epsilon)$ be the time evolution matrix. Then it follows that
    \begin{align*}
        \{1, e^{\pm i \theta(\epsilon)}\} \subset \mathrm{Spec}(U(\epsilon)).
    \end{align*}
    Corresponding eigenvectors $\ket{\psi_1} \in \mathrm{Ker}(I-U(\epsilon))$ and $\ket{\psi_{\pm\theta}} \in \mathrm{Ker}(e^{\pm i\theta(\epsilon)}-U(\epsilon))$ are given by
    \begin{align*}
        \psi_1(a)=\frac{1}{\sqrt{|A_1|+|A_2|}},
    \end{align*} 
    \begin{align*}
        \psi_{\pm\theta}(a)=\frac{1-e^{\pm i\theta(\epsilon)}}{\sqrt{2|A_1||A_2|(|A_1|+|A_2|)}\,|\sin (\theta(\epsilon))|}\times
        \begin{cases}
            -|A_2| &: t(a) \in V_1, \\
            |A_1| &: t(a) \in V_2,
        \end{cases}
    \end{align*}
    respectively.
\end{lemma}
Since $U(\epsilon)$ is the unitary, we have
\begin{align*}
    \mu_{t}(H_1)&=\sum_{a \in A_1}|\bra{a}U(\epsilon)^t\ket{\psi_0}|^2 \\
    &=\sum_{a \in A_1}\left|\sum_{\red{\nu}\in \mathrm{Spec}(U(\epsilon))}\nu^t\braket{a|\psi_{\red{\nu}}}\braket{\psi_{\red{\nu}}|\psi_0}\right|^2.
\end{align*}
\red{
Here, $\psi_{\nu}$ is the eigenvector corresponding to the eigenvalue $\nu$.
}
We focus on the overlap between the eigenvectors and the initial state. Then we get the following lemma.
\begin{lemma}
\label{lem: overlap}
Let $\ket{\psi_1}$ and $\ket{\psi_{\pm \theta}}$ be eigenvectors of $U(\epsilon)$ corresponding to eigenvalues 1 and $e^{\pm i\theta(\epsilon)}$, respectively.
Let $\ket{\psi_0}$ be an initial state given by Eq. (\ref{eq: initial state}). Then we have
    \begin{align*}
        |\braket{\psi_1|\psi_0}|&=\sqrt{\frac{|A_1|}{|A_1|+|A_2|}}, \\
        |\braket{\psi_{\pm\theta}|\psi_0}|&=\frac{1}{\sqrt{2}}\sqrt{\frac{|A_2|}{|A_1|+|A_2|}}+O(\epsilon).
    \end{align*}
\end{lemma}
\begin{proof}
    By directly, we get the first formula. Combining Eq.\,(\ref{eq: initial state}) with Lemma \ref{lem: eigenpair of U}, we see
    \begin{align*}
        |\braket{\psi_{\pm \theta}|\psi_0}|=\frac{1}{\sqrt{2}}\sqrt{\frac{|A_2|}{|A_1|+|A_2|}}\times \left| \frac{1-e^{i\theta(\epsilon)}}{\sin \theta(\epsilon)}\right|.
    \end{align*}
    Thus, we have
    \begin{align*}
        \left| \frac{1-e^{\pm i\theta(\epsilon)}}{\sin \theta(\epsilon)}\right|
        &= \left|\frac{e^{\pm i\theta(\epsilon)/2}(e^{\pm i \theta(\epsilon)/2}-e^{\mp i \theta(\epsilon)/2})}{\sin \theta(\epsilon)} \right| \\
        &= \left|\frac{2i \sin (\theta(\epsilon)/2)}{\sin \theta(\epsilon)}\right| \\
        &= \left|\frac{1}{\cos (\theta(\epsilon)/2)} \right| \\
        &= 1+O(\epsilon).
    \end{align*}
    Therefore, we get the desired conclusion.
\end{proof}
From Lemma \ref{lem: overlap}, $\mu_{t}(H_1)$ is asymptotically described as
\begin{align*}
    &\mu_{t}(H_1) \\&=\sum_{a \in A_1}\left|\braket{a|\psi_1}\braket{\psi_1|\psi_0}+e^{it\theta(\epsilon)}\braket{a|\psi_{\theta}} \braket{\psi_{\theta}|\psi_0} +e^{-it\theta(\epsilon)}\braket{a|\psi_{-\theta}}\braket{\psi_{-\theta}|\psi_0}\right|^2+O(\epsilon).
\end{align*}
For any $a \in A_1$, the first term of above equation is denoted by
\begin{align*}
\braket{a|\psi_1}\braket{\psi_1|\psi_0}&=
    \frac{1}{\sqrt{|A_1|+|A_2|}}\times \sqrt{\frac{|A_1|}{|A_1|+|A_2|}} \\
    & =\frac{1}{\sqrt{|A_1|}}\frac{|A_1|}{|A_1|+|A_2|}.  \\
\end{align*}
The sum of second and third term of above equation is given by 
\begin{align*}
    &e^{it\theta(\epsilon)}\braket{a|\psi_{\theta}} \braket{\psi_{\theta}|\psi_0} +e^{-it\theta(\epsilon)}\braket{a|\psi_{-\theta}}\braket{\psi_{-\theta}|\psi_0} \\
    &=-\frac{1}{2\sqrt{|A_1|} \sin^2 \theta(\epsilon)}\frac{|A_2|}{|A_1|+|A_2|} \left(e^{it \theta(\epsilon)}(1-e^{i \theta(\epsilon)})^2+e^{-it \theta(\epsilon)}(1-e^{-i \theta(\epsilon)})^2\right) \\
    &=-\frac{1}{2\sqrt{|A_1|} \sin^2 \theta(\epsilon)}\frac{|A_2|}{|A_1|+|A_2|}\times (-8 \sin^2 (\theta(\epsilon)/2)\cos(t\theta(\epsilon))) \\
    &= \frac{1}{\sqrt{|A_1|}}\frac{|A_2|}{|A_1|+|A_2|}\cos(t \theta(\epsilon))+O(\epsilon).
\end{align*}
Thus, we show
\begin{align*}
    \mu_{t}(H_1)&=\sum_{a \in A_1}\left|\frac{1}{\sqrt{|A_1|}}\frac{|A_1|}{|A_1|+|A_2|}+\frac{1}{\sqrt{|A_1|}}\frac{|A_2|}{|A_1|+|A_2|}\, \cos(t\,\theta(\epsilon))\right|^2+O(\epsilon) \\
    &=\left(\frac{|A_1|+|A_2|\cos (t \theta(\epsilon))}{|A_1|+|A_2|}\right)^2+O(\epsilon).
\end{align*}
Similarly, we have the asymptotic behavior of $\mu_t(H_2)$.
Therefore, we get the desired result.

\section{Summary and discussion}
This paper \red{investigated} a phenomenon called pulsation, inspired by quantum search algorithms and perfect state transfer. We considered \red{a perturbed} Grover walk on a graph formed by connecting two arbitrary graphs via a single edge, referred to as a bridge. A parameter $\epsilon > 0$, representing the strength of connectivity, is assigned to the bridge.
\red{Since the time evolution operator of the perturbed Grover walk is uniquely determined by only the underlying graph, we could focus on the graph structures influence to the pulsation and found the following properties. }
We showed that for sufficiently small $\epsilon$, pulsation \red{occurs} in this setting: that is, the quantum walker periodically transfers between the two graphs. We derived an asymptotic expression for the probability of finding the walker on either graph \red{for $\epsilon \ll 1$}. These results revealed that the behavior of pulsation is determined not by \red{detailed adjacency information of each graph and also the place of the bridge between them}, but solely by the number of arcs in each \red{(Theorem \ref{thm: prob})}.
Furthermore, we demonstrated that the quantum walker goes back and forth between the two graphs, with a period of order $O(\epsilon^{-1/2})$ \red{(Theorem \ref{thm: periodicity})}.
\red{We discovered a universal property of the Grover walk, which is the pulsation. }
\red{However}, in this study, we focused on the case of a composed graph consisting of two graphs connected by one edge. 
\red{Several compelling avenues for further investigation emerge from this work, including the extension of our model to multiple weak bridges or more complex graph compositions.}
\red{Moreover, we have a lot of future's works connecting to the related topics by extending our model.} 

We guess that since {\it a \red{weak} bridge} might be regarded as the potential barrier for \red{small $\epsilon$}, this phenomenon may be interpreted as a model that handles the phenomenon known as a kind of the tunneling effect \cite{B,DH} in quantum mechanics using quantum walks. \red{We need to find its explicit connection in future's work. }

\red{We expect that our model} may be \red{also applied to} an analogy of quantum batteries \cite{AF}. One of the objectives of those studies is to extract the energy from the batteries efficiently by using unitary operators. Whereas, this paper showed that the energy (quantum walker) was extracted from $H_1$ to $H_2$ by $U(\epsilon)$. We conceive that the results of this paper may be one of the directions of quantum batteries and quantum walks on finite graphs.


\red{
In contrast to previous studies \cite{CWZ,Chiral}, where probability transfer was achieved by modulating the time-reversal symmetry, the present study maintains the time-reversal symmetry of the evolution matrix $U(\epsilon)$. Instead, we focus on a quantum walk where the initial state is localized solely on one side $H_{1}$. This breaks the symmetry about the initial state and may be seen as one factor enabling probability transport. However, since the dynamics of the quantum walk become excessively complex in the absence of a weak bridge (the $\epsilon=1$ case), the introduction of this bridge is considered another essential factor. A more detailed investigation into the interplay between various symmetries and the pulsation remains a subject for future work.
}

\red{
It is known that several connections between quantum walks and electric circuits have already been established. For example, previous researches \cite{Ecircuit, PCSZ} has shown that some class of quantum walks can be realized based on the electrical circuits, both theoretically and experimentally. Furthermore, the papers \cite{HSS1, HS2} have already shown that Grover walk on graphs with tails can be described by circuit equations which are extended electrical circuit equations. Our paper has not yet explicitly revealed such relationships between the circuits equation and quantum walk. However, we think that extending the analysis to cases with multiple weak edges may have a connection between circuit equations and the behavior of quantum walks. These relationships could potentially suggest that some phenomena, such as state transfer or quantum search in quantum walks, might be achievable using classical electrical circuits. Clarifying the connection between the electric circuit and this work is one of the direction of the future works.
}

\section*{Acknowledgments}
E.S. acknowledges financial supports from the Grant-in-Aid of Scientific Research (C) Japan Society for the Promotion of Science (Grant No. 24K06863).
The authors have no competing interests to declare that are relevant to the content of this article. All data generated or analyzed during this study are included in this published article.
\bibliographystyle{plain}
\bibliography{ref}
\end{document}